\documentclass[%
aps,pra%
]{revtex4-2}
\usepackage{graphicx}
\usepackage{dcolumn}
\usepackage{bm}
\usepackage{xcolor}
\usepackage{hyperref} %
\usepackage%
{hyperref} %
\hypersetup{
    colorlinks,
    linkcolor={blue!80!black},%
    citecolor={green!30!black},
    urlcolor={blue!80!black}
}

\usepackage{physics}

\usepackage{amsthm}
\usepackage{mathrsfs}
\usepackage{bbm} 

\usepackage{comment}

\usepackage{accents}
\newlength{\dhatheight}
\newcommand{\doublehat}[1]{%
    \settoheight{\dhatheight}{\ensuremath{\hat{#1}}}%
    \addtolength{\dhatheight}{-0.35ex}%
    \hat{\vphantom{\rule{2pt}{\dhatheight}}%
    \smash{\hat{#1}}}}

\newcommand{\identity}{\mathbbm{1}}
\renewcommand{\trace}{\mathrm{Tr}}

\theoremstyle{definition}
\newtheorem{theorem}{Theorem}
\theoremstyle{definition}	
\newtheorem{lemma}[theorem]{Lemma}
\theoremstyle{definition}
\newtheorem{corollary}[theorem]{Corollary}
\theoremstyle{definition}	

\theoremstyle{definition} 
\newtheorem{definition}{Definition}
\theoremstyle{definition}
\newtheorem{remark}{Remark}
\theoremstyle{definition}

\begin{document}

\preprint{APS/123-QED}

\title[Communication Over Entanglement-Breaking Channels With Unreliable Entanglement Assistance]{Communication Over Entanglement-Breaking Channels\\ With Unreliable Entanglement Assistance}

\author{Uzi Pereg}
 \altaffiliation[Also at ]{Helen Diller Quantum Center, Technion}
 \email{uzipereg@technion.ac.il}
\affiliation{%
 ECE Department, Technion.
}%

\date{\today}

\begin{abstract}
Entanglement assistance can improve communication rates significantly. Yet, its generation is susceptible to failure. The  unreliable assistance  model accounts for those challenges. Previous work provided an asymptotic formula that outlines the tradeoff between the unassisted and excess rates from entanglement assistance. We derive a full characterization for entanglement-breaking channels, and show that combining entanglement-assisted and unassisted coding is suboptimal. From a networking perspective, this finding is nontrivial and highlights a quantum behavior arising from superposition.

\end{abstract}

\maketitle


\section{\label{sec:Introduction}Introduction}

Quantum entanglement has the potential to revolutionize communication systems, as it could be used to transmit information at speeds far beyond what is possible classically 
\cite{LXLCLWYWSL:23p,WangRahman:22p,HaoShiLiShapiroZhuangZhang:21p}.
 In optical communications, generating pre-shared entanglement between the transmitter and the receiver can be challenging due to photon absorption during transmission. Therefore, practical systems rely on a back channel to confirm successful entanglement generation \cite{ShchukinSchmidtvanLoock:19p}. However, this introduces delays and further degrades entanglement resources. The author, along with Deppe and Boche \cite{PeregDeppeBoche:23p1}, proposed an alternative approach for communication with unreliable entanglement assistance. Our principle of operation provides reliability by design, by  adapting the communication rate based on the availability of entanglement assistance, while eliminating the need for feedback, repetition, or distillation.

A fundamental task in information theory is to
determine the channel capacity, i.e., the ultimate transmission rate of communication with a vanishing probability of decoding error.
 The Holevo-Schumacher-Westmoreland (HSW) Theorem provides an
 asymptotic description of the capacity of a quantum channel in the form of a multi-letter regularized expression
\cite{Holevo:98p,SchumacherWestmoreland:97p}. 
One may employ the HSW theorem to  compute lower bounds on the capacity and even obtain a complete characterization in specific examples. 
However, in Shannon Theory, 
multi-letter capacity formulas are generally considered an incomplete solution, for reasons of computability 
 \cite{Korner:87b}, uniqueness \cite{Wilde:17b}, and  insights on optimal coding 
 \cite{ElGamalKim:11b}.
%
%
In the entanglement-assisted communication setting, where
pre-shared entanglement resources are available to the transmitter and the receiver, 
 a complete single-letter characterization is well established \cite{BennettShorSmolin:02p} and can be viewed as the quantum parallel of Shannon's capacity theorem \cite{Shannon:48p}.
 Therefore, entanglement-assisted communication has favorable attributes from both performance and analysis perspectives.


Let us now consider communication with \emph{unreliable} entanglement assistance.
Suppose that
Alice wishes to send two messages, at rates $R$ and $R'$.
  She encodes both messages using her share of the entanglement resources, as she does not know whether Bob will have access to the entangled resources. Nevertheless, heralded entanglement generation guarantees that Bob knows whether the procedure was successful or not. Bob has two decoding procedures. If the entanglement assistance has failed to reach Bob's location, he performs a decoding operation to recover the first message alone. Hence, the communication system operates on a rate $R$. Whereas if Bob has entanglement assistance, he decodes both messages, hence the overall transmission rate is $R+R'$.
  In other words, $R$ is a guaranteed rate, and $R'$ 
  is the excess rate of information that entanglement assistance provides.

The previous work \cite{PeregDeppeBoche:23p1} established an asymptotic regularized formula for the capacity region, i.e., the set of all rate pairs $(R,R')$ that can be achieved with a vanishing probability of decoding error. The achievability scheme is  inspired by the classical network technique of superposition coding (SPC).
We refer to the quantum method as \emph{quantum SPC}.
The classical technique consists of layered codebooks, by which the codewords are divided into so-called cloud centers and satellites, representing the first and second layers, respectively.
 In analogy, quantum SPC uses conditional quantum operations that map quantum cloud centers to quantum satellite states. 
Decoding is performed in two stages. First, Bob recovers the cloud index, corresponding to the guaranteed information. If the entanglement assistance is absent, then Bob quits after the first step. Otherwise, if Bob has entanglement assistance, then he continues to decode the satellite, i.e., the excess information. 
Until now, it has remained unclear whether quantum SPC is optimal.

Entanglement breaking is a fundamental property of a large class of quantum channels, mapping any entangled state to a separable state \cite{HorodeckShorRuskai:03p}.
 One example is the qubit depolarizing channel, which is  entanglement breaking only when the depolarization parameter is greater than or equal to $2/3$  \cite{MoravvcikovaZiman:10p}.
 From a Shannon-theoretic perspective,  entanglement-breaking channels are much better understood, compared to general quantum channels \cite{Shor:02p,Holevo:08p,WildeHsieh:12p,WildeWinterYang:14p,WangDasWilde:17p,DingWilde:18p}.
 In particular,
 the unassisted capacity is characterized by the single-letter Holevo information \cite{Shor:02p}.
While an entanglement-breaking channel cannot be used to generate entanglement, it may facilitate the transmission of classical messages, and entanglement assistance can increase the channel  capacity for sending classical information substantially \cite{HaoShiLiShapiroZhuangZhang:21p}.
 Shor \cite{Shor:02p} established the single-letter characterization of the unassisted capacity by first showing that the Holevo information of an entanglement-breaking channel is additive.
The author \cite{Pereg:22p} has recently pointed out
 a more direct approach,  proving a single-letter converse proof ``from scratch".

 Entanglement breaking channels and their properties have been extensively studied in the literature \cite{King:03p1,Sacchi:05p,AhiableKribLevickPereiraRahaman:21p,KribsLevickPereiraRahaman:22a,DevendraSapraSumesh:23a}. 
Matsumoto et al. \cite{MatsumotoShimonoWinter:04p}
portrayed the relation between the additivity property and the entanglement of formation.
Wilde et al. \cite{WildeWinterYang:14p} proved the strong converse property for entanglement-breaking channels.
Entanglement breaking multiple-access channels and broadcast channels are considered in 
\cite{GrudkaHorodecki:10p} and \cite{WangDasWilde:17p}, respectively.
More recently, M\"uller-Hermes and Singh \cite{MullerHermesSingh:22a} showed that if the positive partial transposition (PPT) condition holds for both the channel and its complementary, then the channel  is entanglement breaking, and thus anti-degradable (see also \cite{CubittRuskiSmith:08p,HircheLeditzky:22p}).

In this work, we establish full characterization of the capacity region with unreliable entanglement assistance for the class of entanglement-breaking channels. Our main contribution is thus  a converse result that complements the previous achievability proof, and shows that quantum SPC is indeed optimal for the class of entanglement-breaking channels. The analysis relies on observations from another work by the author \cite[Sec. III-D]{Pereg:22p} along with the geometric properties of the rate region.
To complete the characterization, we single-letterize our capacity formula and  show that the auxiliary systems have bounded dimensions.

We also demonstrate
our results for
an entanglement-breaking depolarizing channel.
We show that quantum SPC can outperform time division even in this simple point-to-point setting. This is surprising because SPC is typically useful in more complex network setups, and does not yield an advantage in point-to-point communication.
For example,   in a classical broadcast channel with degraded messages, where a transmitter communicates with two receivers, SPC is unnecessary when the receivers' outputs are identical,  as  
the capacity region can be attained using a simpler approach of time division. 
That is, concatenating two single-user codes is optimal.
In our context, the system can be regarded as a quantum broadcast channel with degraded messages where one receiver has entanglement assistance, and the other does not.  Nevertheless,  the output states of the receivers are identical (without violating the no-cloning theorem, as we consider two alternative scenarios).
The expectation would be that time division, combining assisted and unassisted codes, achieves optimality. However, this expectation is proven false as quantum SPC can outperform time division, based on the combination of a superposition code with a superposition state.

\subsection*{Illustrative Metaphor}
  Communication with unreliable entanglement assistance is not a mere combination of the entanglement-assisted and unassisted settings.
The protocol poses a challenge as Alice must encode without knowledge of the availability of assistance. The availability of entanglement is not associated with a probabilistic model either.
To illustrate the concept of reliability,  consider the following metaphor. 

Imagine there are $N$ travelers embarking on a journey aboard a ship that may have a variable number of lifeboats. The total capacity of the lifeboats is $L$, which determines how many travelers can be accommodated in case of a shipwreck, $L\leq N$. The ship's speed is denoted as $V\equiv V(N,L)$, while the lifeboats' speed is  $v_0$. If the ship does not sink, each traveler will travel at speed $V$. 
To avoid a morbid narrative, let us envision that in the event of an unforeseen shipwreck, $(N-L)$ travelers will be safely rescued and brought back to the starting point, while the journey continues with the remaining travelers aboard the lifeboats.
 The speed of travel in this scenario is calculated as the average speed of the lifeboats, $R=(L/N)v_0$. 
 
In our metaphor, $R$ represents the guaranteed speed for the remaining travelers, while $R'=V-R$ indicates the excess speed that the ship would have provided. Increasing the number of lifeboats improves the guaranteed speed but reduces the excess speed, while decreasing the number of lifeboats has the opposite effect. When planning for the worst-case scenario, it is crucial to consider both speeds, $R$ and $R'$, rather than just the average speed.

One may consider  the option of dividing the travelers among a heavy ship and a light ship.
Figuratively, our findings show
that if the journey is subject to a quantum evolution, then we
may outperform the division plan by allowing travelers to be
in a quantum superposition state between the two ships.

\section{Coding with Unreliable Assistance}
\label{subsec:Mcoding}

\subsection{Notation,  Information Measures, and Quantum Channels}
\label{Subsec:Notation}
We use standard notation for quantum channels and information measures,  as in
\cite[Chap. 11]{Wilde:17b}.
The letters $X,Y,Z,\ldots$ represent discrete random variables, on finite sets   $\mathcal{X},\mathcal{Y},\mathcal{Z},...$, respectively. 
 The distribution of  $X$ is specified by a probability mass function (pmf) 
	$p_X(x)$ on $\mathcal{X}$. 
 We use $x^n=(x_i)_{i\in [n]}$ to denote  a sequence of letters from $\mathcal{X}$.

The state of a quantum system $A$ is given by a density operator on the Hilbert space $\mathcal{H}_A$.
A measurement  is specified by a collection of operators $\{D_j \}$ that forms a positive operator-valued measure (POVM), i.e.,
 $D_j\geq 0$  and  
$\sum_j D_j=\identity$, where $\identity$ is the identity operator.
Given a bipartite state $\rho_{AB}$, 
define the quantum mutual information by
$
I(A;B)_\rho=H(\rho_A)+H(\rho_B)-H(\rho_{AB}) 
$, 
where $H(\rho) \equiv -\trace[ \rho\log(\rho) ]$ is the von Neumann entropy.
The conditional quantum entropy and mutual information are defined by
$H(A|B)_{\rho}=H(\rho_{AB})-H(\rho_B)$ and
$I(A;B|C)_{\rho}=H(A|C)_\rho+H(B|C)_\rho-H(A,B|C)_\rho$, respectively.

A quantum channel $\mathcal{N}_{A\to B}$ is a completely-positive trace-preserving (cptp) map. 
If the systems $A^n=(A_1,\ldots,A_n)$ are sent through $n$ channel uses, then the input state $\rho_{A^n}$ undergoes the tensor product mapping $\mathcal{N}_{A\to B}^{\otimes n}$.
The channel  is called entanglement breaking if for every  input state $\rho_{A A'}$, where $A'$ is an arbitrary reference system, the channel output is separable,  i.e.,
$
( \mathcal{N}_{A\rightarrow B} \otimes \identity)(\rho_{A A'})= \sum_{x\in\mathcal{X}} p_X(x)  \psi_{B}^x \otimes \psi^x_{A'} 
$, 
for some pmf $p_X$ and pure states $\psi_{B}^x$ and $\psi_{A'}^x$.  
The Kraus representation of an entanglement breaking channel consists of unit-rank Kraus operators. 
Furthermore, every entanglement-breaking channel  can be represented as a serial concatenation of a measurement channel followed by a classical-quantum channel \cite[Corollary 4.6.1]{Wilde:17b}. 

\subsection{Coding and Channel Capacity}
We define a  code for communication with unreliable entanglement resources.
 Alice and Bob's entangled systems are denoted 
$T_A$ and $T_B$, respectively.

\begin{definition} 
A $(2^{nR},2^{nR'},n)$   code with unreliable entanglement assistance consists of the following:   
Two message sets $[2^{nR}]$ and $[2^{nR'}]$, where $2^{nR}$, $2^{nR'}$ are  integers, an entangled state $\Psi_{T_A,T_B}$, 
  a collection of encoding maps $\mathcal{F}^{m,m'}_{T_{A}\rightarrow A^n}
  $ for $m\in [2^{nR}]$ and $m'\in [2^{nR'}]$,  and two decoding POVMs,
	$\mathcal{D}_{B^n T_B}
	=\{ D_{m,m'} \}$ and $\mathcal{D}^*%
	_{B^n}=\{ D^*_{m} \}$.
%
\end{definition} 
Alice chooses two  messages, $m\in [2^{nR}]$ and $m'\in [2^{nR'}]$. %
She applies the encoding map 
to her share of the entangled state, and then transmits %
$A^n$ over $n$ channel uses of $\mathcal{N}_{A\rightarrow B}$. 
 Bob receives 
$B^n$. 
If the entanglement assistance is present,  i.e., Bob has access to the  resource $T_B$, then he should recover both messages. He  performs a joint measurement
 $\mathcal{D}_{B^n T_B}$ to obtain an estimate $(\hat{m},\hat{m}')$.

Otherwise, if entanglement assistance is absent,  Bob does not have $T_B$. 
Hence,
 he performs the measurement $\mathcal{D}^*_{B^n}$ to obtain an estimate $\doublehat{m}$ of the first message alone.
 The error probability   is  
\begin{align}
&P_{e|  m,m'}^{(n)}
= 1- %
\trace\big[ \mathcal{D}\circ
\mathcal{N}^{\otimes n}_{A\rightarrow B}\circ
\mathcal{F}^{m,m'} (\Psi_{T_A,T_B})
 \big] %
\intertext{in the presence of entanglement assistance, and }
&P_{e|  m,m'}^{*(n)}= 
1-
\trace\Big[ \mathcal{D}^*\circ
\mathcal{N}^{\otimes n}_{A\rightarrow B}\circ
\mathcal{F}^{m,m'}(\Psi_{T_A})  \Big] 
\end{align}
without assistance.
The encoded input remains the same in both scenarios,  since Alice does not know whether entanglement is available or not. Therefore, the  error depends on $(m,m')$ in both cases.
A rate pair $(R,R')$ is  achievable  if there exists a sequence of $(2^{nR},2^{nR'},n)$ 
codes with unreliable entanglement assistance, such that 
$\max(P_{e|  m,m'}^{(n)},P_{e|  m,m'}^{*(n)} )\to 0$ as $n\to\infty$. 
The  capacity region $\mathcal{C}_{\text{EA*}}(\mathcal{N})$ with unreliable entanglement assistance is defined as the set of  achievable rate pairs.

\section{Results}
Let $\mathcal{N}_{ A\rightarrow B}$ be an entanglement-breaking channel (see Section~\ref{Subsec:Notation}). 
Define the region
\begin{align}
\label{eq:calRClea}
\mathcal{R}_\text{EA*}(\mathcal{N})
=
\bigcup
\left\{ \begin{array}{rl}
  (R,R') \,:\;
	R \leq& I(X;B)_\omega  \\
  R'   \leq& I(G_2;B|  X)_\omega
	\end{array}
\right\}
\end{align}
where the union is over all auxiliary variables $X\sim p_X$, all quantum states $\varphi_{G_1 G_2}$, and all encoding channels $\mathcal{F}^{(x)}_{G_1\to A}$,
\begin{align}
\omega_{XA G_2}
&=
\sum_{x\in\mathcal{X}} p_X(x) \ketbra{x}\otimes (\mathcal{F}^{(x)}_{G_1\to A} \otimes \text{id})
(\varphi_{ G_1 G_2}) \,,  
\nonumber\\
\omega_{XB G_2  }
&=
(\text{id}\otimes\mathcal{N}_{A\rightarrow B}\otimes \text{id})(\omega_{XA  G_2})
\,.
\end{align}
Intuitively,  $X$  represents  the guaranteed information, and $G_1$, $G_2$ are Alice and Bob's resources. Since the entangled resources $G_1$ and $G_2$ are pre-shared, 
the state is 
uncorrelated with the messages.
 Alice encodes the excess information   using the encoding channel $ \mathcal{F}^{(x)} $.

\subsection{Capacity Theorem}

Our main results are stated below, characterizing the capacity region for communication over entanglement-breaking channels with unreliable entanglement assistance. Previous work \cite{PeregDeppeBoche:23p1} established a regularized  characterization for the capacity region, i.e., an asymptotic \emph{multi-letter formula} of the form 
$\bigcup_{K= 1}^{\infty} \frac{1}{K} \mathcal{R}(\mathcal{N}^{\otimes K})$.
Here, 
we provide a complete characterization in the form of a single-letter formula.


\begin{theorem}
\label{theo:ClEA}
The capacity region of an entanglement-breaking quantum channel $\mathcal{N}_{A\to B}$ with unreliable entanglement assistance is given by
\begin{align}
\mathcal{C}_\text{EA*}(\mathcal{N})= \mathcal{R}_\text{EA*}(\mathcal{N})
\end{align}
where $\mathcal{R}_\text{EA*}(\mathcal{N})$ is as defined in \eqref{eq:calRClea}.
\end{theorem}
The proof is given in Section~\ref{app:capacity}.

\begin{remark}
\label{Remark:Single_Letter}
Single-letterization is highly valued in Shannon theory for reasons of computability 
 \cite{Korner:87b}, uniqueness \cite{Wilde:17b}, and  insights on optimal coding 
 \cite{ElGamalKim:11b}. See further discussion in Section~\ref{Section:Single_Letterization}.
However, the result in Theorem~\ref{theo:ClEA} in itself is not enough to claim that this is truly a single-letter characterization, as 
  the computation of a rate region requires specified dimensions. 
 Thereby, we show in Section~\ref{Subsection:Single_Letter} that the 
 auxiliary systems, $X$, $G_1$, and $G_2$, all 
 have bounded dimensions.
Together, the results in Theorem~\ref{theo:ClEA} and Section~\ref{Subsection:Single_Letter} complete the characterization.  
\end{remark}

\subsection{Equivalent characterization}
\label{Subsection:Equivalence}
  Before we prove the capacity theorem, we establish useful properties of the region 
  $\mathcal{R}_\text{EA*}(\mathcal{N})$,  as defined in \eqref{eq:calRClea}.
We  
show an equivalence to the  region below:
\begin{align}
\mathcal{O}_\text{EA*}(\mathcal{N})
&=
\bigcup
\left\{ \begin{array}{rl}
  (R,R') \,:\;
	R \leq& I(X;B)_\omega  \\
  R+R'   \leq& I(X G_2;B)_\omega
	\end{array}
\right\}
\label{eq:calRClea2}
\end{align}
where the union is  as in \eqref{eq:calRClea}.
This  will be useful in the  proof for our main theorem in Section~\ref{app:capacity}, where we will show that every achievable rate pair must lie within $\mathcal{O}_\text{EA*}(\mathcal{N})$.

We give the intuition below.
In order to show the equivalence between the regions $\mathcal{R}_\text{EA*}(\mathcal{N})$ and $\mathcal{O}_\text{EA*}(\mathcal{N})$, we use the  geometric properties of our regions, as illustrated in 
 Figure~\ref{fig:Convexity}.
 The region $\mathcal{R}_\text{EA*}(\mathcal{N})$ is defined in \eqref{eq:calRClea} as a union of rectangles.
In particular, the light shaded rectangle in Figure~\ref{fig:Convexity} corresponds to
  the bounds $0\leq R\leq I(X;B)_\omega$ and 
$0\leq R'\leq I(G_2;B|X)_\omega$,  for a fixed auxiliary variable, state, and encoding channels.
The corner point of the region is denoted by 
$P_0=\left(I(X;B)_\omega ,I(G_2;B|X)_\omega\right)$.
Similarly, the region $\mathcal{O}_\text{EA*}(\mathcal{N})$ is a union of trapezoids, the corners of which are
$P_0$ and $P_1=\left(0 ,I(X G_2;B)_\omega\right)$.
Hence, the dark shaded area in Figure~\ref{fig:Convexity} is the gap between the rectangle and the trapezoid, which the regions  $\mathcal{R}_\text{EA*}(\mathcal{N})$ and $\mathcal{O}_\text{EA*}(\mathcal{N})$ comprise of. 
Now, observe that the point $P_1$ belongs to \emph{another} rectangle in $\mathcal{R}_\text{EA*}(\mathcal{N})$, taking $\bar{G}_2$ and $\bar{X}$ to be $(X,G_2)$ and null, respectively. 
Therefore,
the convexity of the region $\mathcal{R}_\text{EA*}(\mathcal{N})$ implies that any convex combination $P_\lambda=(1-\lambda)P_0+\lambda P_1$ must also lie within $\mathcal{R}_\text{EA*}(\mathcal{N})$, thereby $\mathcal{R}_\text{EA*}(\mathcal{N})=\mathcal{O}_\text{EA*}(\mathcal{N})$. The details are given below.

\begin{figure}[tb]
\center
\includegraphics[scale=1,trim={2cm 18cm 1cm 4cm},clip]
{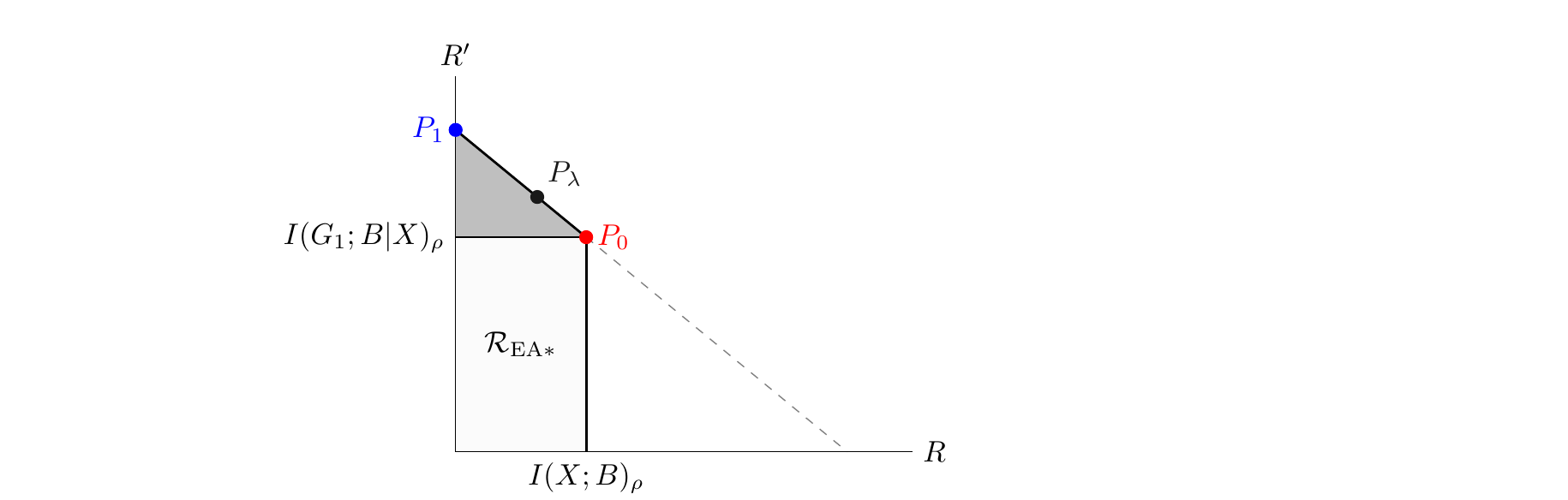} 
\caption{Achievable rate regions. 
}
\label{fig:Convexity}
\end{figure}

We begin with the convexity of 
 our original region. 
\begin{lemma}
\label{Lemma:Convexity}
The rate region $\mathcal{R}_\text{EA*}(\mathcal{N})$ is a convex set.
\end{lemma}
As explained above,  the convexity of $\mathcal{R}_\text{EA*}(\mathcal{N})$ implies that the point $P_\lambda$ in Figure~\ref{fig:Convexity} is included within the union of rectangles, i.e., $P_\lambda\in \mathcal{R}_\text{EA*}(\mathcal{N})$. We obtain the following consequence.
\begin{corollary}
\label{Corollary:Region_lambda}
For every $\lambda\in [0,1]$,
\begin{align}
\mathcal{R}_\text{EA*}(\mathcal{N})
\supseteq
%
\left\{ \begin{array}{rl}
  (R,R') \,:\;
	R \leq& (1-\lambda)I(X;B)_\omega  \\
    R'\leq& I(G_2;B|  X)_\omega+\lambda I(X;B)_\omega
	\end{array}
\right\} \,.
\label{eq:calRClea3}
\end{align}
\end{corollary}
The proof for the convexity properties in Lemma~\ref{Lemma:Convexity} and Corollary~\ref{Corollary:Region_lambda}  is given in Appendix~\ref{app:Convexity}. 
Next, we use those properties  to establish equivalence. 


 \begin{lemma}[Equivalence]
 \label{Lemma:Equivalence}
$
\mathcal{R}_\text{EA*}(\mathcal{N})=\mathcal{O}_\text{EA*}(\mathcal{N}) 
$. 
 \end{lemma}

\begin{proof}
The inclusion $\mathcal{R}_\text{EA*}(\mathcal{N}) \subseteq \mathcal{O}_\text{EA*}(\mathcal{N})$ is immediate 
by the chain rule.
 It remains to show that 
every rate pair in the region $\mathcal{O}_\text{EA*}(\mathcal{N})$, belongs to 
$\mathcal{R}_\text{EA*}(\mathcal{N})$ as well.

Let $(R,R')\in \mathcal{O}_\text{EA*}(\mathcal{N})$, hence
 \begin{align}
 R&\leq I(X;B)_\omega 
 \label{Equation:Outer_Guranteed_Rate}
 \,,\; 
 R+R'\leq I(XG_2;B)_\omega \,.
 \end{align}
By the first inequality, there exists $0\leq\lambda\leq 1$ such that
\begin{align}
 R&=(1-\lambda)I(X;B)_\omega \,.
 \label{Equation:R_bar}
\intertext{%
 By \eqref{Equation:Outer_Guranteed_Rate}
 -\eqref{Equation:R_bar}, 
 }
 R'&\leq I(XG_2;B)_\omega-R
 \nonumber\\
 &= 
 I(XG_2;B)_\omega-I(X;B)_\omega+\lambda I(X;B)_\omega
 \nonumber\\
 &= I(G_2;B|X)_\omega+\lambda I(X;B)_\omega \,.
 \end{align}
 Hence, by Corollary~\ref{Corollary:Region_lambda},
 $(R,R')\in \mathcal{R}_\text{EA*}(\mathcal{N})$. 
\end{proof}

 \subsection{Single-letterization}
 \label{Subsection:Single_Letter}
 As mentioned above, Single-letterization is highly valued in Shannon theory (see  Remark~\ref{Remark:Single_Letter} and discussion in Section~\ref{Section:Single_Letterization}).
 We establish that our characterization 
 is a single-letter formula. Specifically, the 
 auxiliary systems, $X$, $G_1$, and $G_2$, all 
 have bounded dimensions. 
 Denote the channel input dimension by 
 $d_A\equiv \mathrm{dim}(\mathcal{H}_A) $.
 
\begin{lemma}
\label{lemm:pureCea}
The union in (\ref{eq:calRClea}) is exhausted by pure states $\ket{\phi_{G_1 G_2}}$, cardinality  %
$\abs{\mathcal{X}}\leq d_A^{\,2}+1$,
and dimensions
$\mathrm{dim}(\mathcal{H}_{G_1})=\mathrm{dim}(\mathcal{H}_{G_2})\leq d_A( d_A^{\,2}+1)$. 
\end{lemma}
The first part 
has already been stated in \cite{PeregDeppeBoche:23p1}. The quantum dimension bound is new, see proof  in Section~\ref{Section:Single_Letter_Proof} below.

\section{Analysis}

\subsection{Single-Letterization}
\label{Section:Single_Letter_Proof}
The first part of Lemma~\ref{lemm:pureCea} has already been established in our previous work 
\cite[Lemma 4]{PeregDeppeBoche:23p1},
using convex analysis.
Bounding the quantum dimensions 
is more challenging. 

Consider a pure state,
$\ket{\psi_{G_1 G_2}}$.
%
Since the Schmidt rank 
is bounded by each dimension, we may assume w.l.o.g. that $G_1$ and $G_2$ are qudits of the same  dimension
$
d_0$, for some  $d_0>0$.
We would like to show that the union can be restricted such that encoded state 
$\omega_{G_2 A}^x\equiv (\mathrm{id}\otimes\mathcal{F}_{G_1 A}^{(x)})(\ketbra{\psi_{G_2 G_1}})$ remains pure.

Since every quantum channel has a Stinespring dilation, there exists
a unitary $V^{(x)}$ such that
$
\mathcal{F}_{G_1\to A}^{(x)}(\rho)=
\trace_{D E}\left[ V^{(x)} (\ketbra{0}_D\otimes\rho) V^{(x)\dagger} \right]
$, 
where  $V^{(x)}$  maps from $\mathcal{H}_{D}\otimes \mathcal{H}_{G_1} $ to $\mathcal{H}_{E}\otimes \mathcal{H}_{A} $, while $D, E$ are reference systems with appropriate dimensions.
Since $G_1$ is an arbitrary ancilla, we may include the reference $D$ within this ancilla, and simplify 
as
$
\mathcal{F}_{G_1\to A}^{(x)}(\rho)=
\trace_{ E}\left[ U^{(x)} \rho U^{(x)\dagger} \right]
$, 
where $U^{(x)}$ is a unitary   from $\mathcal{H}_{G_1} $ to $\mathcal{H}_{E}\otimes \mathcal{H}_{A} $. 

We would like the ancilla $G_2$  to absorb the reference  $E$ as well.  Seemingly,  this would contradict \eqref{eq:calRClea} as $E$ could be correlated with $x$.
To resolve this difficulty, we 
show that the encoding operation can be reflected to $G_2$. 
Fix $x\in\mathcal{X}$ and  consider the purification
$
\ket{\omega^{(x)}_{G_2 E A}} \equiv
(\identity\otimes U^{(x)})\ket{\psi_{G_2 G_1}} 
$. 

Let $W_{i,j}$ denote the Weyl operators on 
$\mathcal{H}_{G_1}\cong \mathcal{H}_{G_2}$,
for $i,j\in\{0,\ldots,d_0-1\}$ 
\cite[Sec. 3.7.2]{Wilde:17b}.
By plugging a decomposition of 
$\ket{\psi_{G_2 G_1}}$ in the generalized Bell basis \cite[Ex. 3.7.11]{Wilde:17b},
and applying
 the mirror lemma, by which  $(\identity\otimes U)\ket{\Phi}=(U^T\otimes\identity)\ket{\Phi}$ for every qudit operator $U$
\cite[Ex. 3.7.12]{Wilde:17b},
we obtain
$
\ket{\omega^{(x)}_{G_2 E A}} 
= \sum_{i,j=0}^{d_0-1} \alpha_{i,j}
\left(W_{i,j} F_{G_1\to G_2 E}^{(x)} \otimes \identity_A\right)\ket{\Phi}_{G_1 A} 
$, 
with
$F_{G_1\to G_2 E}^{(x)}=(U^{(x)})^T$.
We see that 
\eqref{eq:calRClea} 
can thus be represented as a union over all unitaries 
$F^{(x)}_{G_1\to G_2 E}\otimes \identity_A$.

In this formulation, both $E$ and $G_2$ are encoded by an operation depending on $x$.
Thus, we can extend the union to 
$\bar{G}_2=(G_2,E)$.
The bound on the guaranteed rate $R$ remains. 
As for the excess rate, 
$I(\bar{G}_2;B|X)_\omega\geq I(G_2;B|X)_\omega$.
Hence, 
it suffices to consider pure states $\ket{\omega^{(x)}_{G_2 A}}$, the Schmidt rank of which is bounded by $d_A$.
Thus, the region is exhausted with $d_0\leq |\mathcal{X}|d_A$.
%
%
%
%
%
%
%
%
%
%
%
\qed

\subsection{Capacity Proof}
\label{app:capacity}
The direct part was proved in our earlier work \cite{PeregDeppeBoche:23p1}. 
We now focus our attention on the converse.
Suppose that Alice and Bob share an unreliable 
resource $\Psi_{T_A T_B}$. Alice first prepares classical correlation,
\begin{align}
\pi_{K M K'M'} \equiv 
\left(\frac{1}{2^{nR}}\sum_{m=1}^{2^{nR}} \ketbra{ m } \otimes \ketbra{ m }\right) 
\otimes
\left(\frac{1}{2^{nR'}}\sum_{m'=1}^{2^{nR'}} \ketbra{ m' } \otimes \ketbra{ m' }\right) %
\end{align}
locally. 
She encodes by  $\mathcal{F}_{M M'T_A \rightarrow A^n}$, and transmits $A^n$.
 Bob receives $B^n$ in the state
$
\omega_{K K'  T_B B^n}\equiv (\text{id}\otimes\mathcal{N}^{\otimes n}\mathcal{F}) (\pi\otimes \Psi) 
$. 
He decodes with either $\mathcal{D}_{B^n T_B\rightarrow \hat{M} \hat{M}'}$ or $\mathcal{D}^*_{B^n \rightarrow \tilde{M} }$, depending on the availability of entanglement assistance.

Consider a sequence of codes $(\mathcal{F}_n,\Psi_n,\mathcal{D}_n,\mathcal{D}_n^*)$ 
with 
vanishing errors.
By continuity and data processing arguments \cite[App. C]{PeregDeppeBoche:23p1},
%
\begin{align}
nR
&\leq I(K;B^n)_{\omega}+n\varepsilon_n^* \,,
\label{eq:ConvIneq1noEA}
\\
n(R+R')&\leq I(K K'T_B; B^n)_{\omega}+n\varepsilon_n 
\label{eq:ConvIneq1noEAp}
\end{align}
where $\varepsilon_n,\varepsilon_n^*\to 0$ as $n\to \infty$.

Since the channel is entanglement-breaking, it can be represented by a  measurement channel $\mathcal{M}_{A\to Y}$, 
followed by a preparation channel $\mathcal{P}_{Y\to B}$, where $Y$ is classical \cite[Sec. III-D]{Pereg:22p}.
%
Define the sequence of classical variables,
$
X_i\equiv (K,Y^{i-1})
$, 
for $i\in [n]$.
By the chain rule and the data processing inequality, \eqref{eq:ConvIneq1noEA}-\eqref{eq:ConvIneq1noEAp} imply
\begin{align}
n(R-\varepsilon_n^*)
&\leq
\sum_{i=1}^n I(K B^{i-1}; B_i)_\omega
\nonumber\\
&\leq
\sum_{i=1}^n I(K Y^{i-1}; B_i)_\omega
\nonumber\\
&=
\sum_{i=1}^n I(X_i; B_i)_\omega \,,
\intertext{and similarly,}
n(R+R'-\varepsilon_n)
&\leq
\sum_{i=1}^n I(K' T_B X_i; B_i)_\omega \,.
\end{align}
Letting $J$ be uniformly distributed index in $[n]$, we have
$
R-\varepsilon_n^*
\leq
I(X_J; B_J|J)_\omega
\leq I(J X_J; B_J)_\omega
$ and %
$
R+R'-\varepsilon_n
\leq
I(K' T_B J X_J;B_J)_\omega
$ 
with respect to 
$
\omega_{J K'T_B X_J B_J}\equiv 
\frac{1}{n}\sum_{i=1}^n \ketbra{i}_J\otimes 
\omega_{K' T_B X_i B_i}
$. 

Taking $G_2\equiv (K',T_B)$, $X\equiv (J,X_J)$, $A\equiv A_J$,  hence
$B\equiv B_J$, we deduce that 
$(R,R')\in\mathcal{O}_{EA*}(\mathcal{N})$.
This, in turn, implies 
 $(R,R')\in\mathcal{R}_{EA*}(\mathcal{N})$, by Lemma~\ref{Lemma:Equivalence}. 
%
%
%
%
%
%
%
%
%
%
%
%
%
\qed %

\section{Example}
Consider the qubit depolarizing channel,
$
\mathcal{N}(\rho)=(1-\varepsilon)\rho+\varepsilon \frac{\identity}{2} 
$, 
with $\varepsilon\in [0,1]$. The unassisted capacity,  $C(\mathcal{N})%
$,  is achieved with a symmetric distribution over  $\{ |0\rangle, |1\rangle \}$ 
(see \cite{King:03p}).
On the other hand, the  capacity with reliable entanglement assistance   $C_{\text{EA}}(\mathcal{N})
$ is achieved with an EPR state \cite{BennettShorSmolin:99p}.
A classical mixture of those strategies yields the time division region,
$
\mathcal{C}_\text{EA*}(\mathcal{N})
\supseteq
\bigcup_{0\leq \lambda\leq 1}
\left\{ \begin{array}{rl}
  (R,R') \,:\;
	R \leq& (1-\lambda)\, C(\mathcal{N})   \\
  R'   \leq& \lambda C_{\text{EA}}(\mathcal{N})
	\end{array}
\right\} 
$. 
We claim that this is suboptimal.

Figure~\ref{fig:clEAur} depicts the capacity region for a parameter such that the channel is entanglement breaking, 
$\varepsilon=0.7$ (as opposed to \cite[Example 1]{PeregDeppeBoche:23p1}).
The time-division bound is below the red line, 
whereas  the  blue curve indicates the capacity region that is achieved using a superposition state. 
\begin{figure}[tb]
\center
\includegraphics[scale=0.515,trim={4.5cm 8.72cm 1cm 8.925cm},clip]
{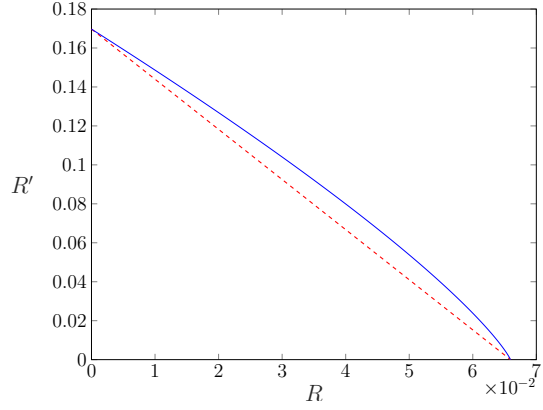} 
\caption{Achievable rate regions. 
}
\label{fig:clEAur}
\end{figure}
Based on Theorem~\ref{theo:ClEA}, we establish that the capacity region of an entanglement-breaking qubit depolarizing channel with unreliable entanglement assistance is given by 
\begin{align}
\mathcal{C}_\text{EA*}(\mathcal{N})
=
\bigcup_{0\leq \alpha\leq \frac{1}{2}}
\left\{ \begin{array}{rl}
  (R,R') \,:\;
	R &\leq 1-h_2\left(\alpha*\frac{\varepsilon}{2} \right)   \\
  R'   &\leq h_2(\alpha)+h_2\left(\alpha*\frac{\varepsilon}{2} \right)-H\bigg(\frac{\alpha \varepsilon}{2},\frac{(1-\alpha)\varepsilon}{2},
  \\
  &\frac{1}{2}-\frac{\varepsilon}{4}-\sqrt{
  \frac{\varepsilon^2}{16}-(1-\alpha)\alpha \varepsilon(1-\frac{3\varepsilon}{4})+\frac{1-\varepsilon}{4}
  },
  \\
  &\frac{1}{2}-\frac{\varepsilon}{4}+\sqrt{
  \frac{\varepsilon^2}{16}-(1-\alpha)\alpha \varepsilon(1-\frac{3\varepsilon}{4})+\frac{1-\varepsilon}{4}
  } \bigg)
	\end{array}
\right\} 
\label{Equation:Depolarizing_Region}
\end{align}
where
$H(\mathbf{p})\equiv -\sum_i p_i\log(p_i)$ is the Shannon entropy for a classical probability vector $\mathbf{p}$, the binary entropy function is denoted by
$h_2(x)\equiv H(x,1-x)$ for $x\in [0,1]$, and $\alpha*\beta=(1-\alpha)\beta+\alpha(1-\beta)$ is the binary convolution operation.

\begin{proof}
By Theorem~\ref{theo:ClEA},
it suffices to evaluate the region $\mathcal{R}_\text{EA*}(\mathcal{N})$, as defined in \eqref{eq:calRClea}.

We begin with the converse part and show that the 
set on the right-hand side of 
\eqref{Equation:Depolarizing_Region} is an outer bound on $\mathcal{R}_\text{EA*}(\mathcal{N})$.
Consider a rate pair $(R,R')\in \mathcal{R}_\text{EA*}(\mathcal{N})$. Hence, 
$R\leq I(X;B)_\omega$ and $R'\leq I(G_2;B|X)_\omega$,
or, equivalently,
\begin{subequations}
\label{Equation:Depolarizing_Converse_1}
\begin{align}
R&\leq H(B)_\omega-H(B|X)_\omega \,,
\\
R'&\leq H(G_2|X)_\omega+ H(B|X)_\omega-H(G_2 B|X)_\omega \,,
\end{align}
\end{subequations}
for some pure input state $\ket{\phi_{G_1 G_2}}$,
variable $X\sim p_X$, and
encoder 
$\mathcal{F}_{G_1\to A}^{(x)}$ (see Lemma~\ref{lemm:pureCea}).

Based on the analysis in Section~\ref{Section:Single_Letter_Proof},
it suffices to consider
an  encoder
that produces a pure state $\ket{\omega_{G_2 A}^{(x)}}$,
for $x\in\mathcal{X}$.
%
Consider a Schmidt decomposition,
\begin{align*}
\ket{\omega_{G_2 A}^{(x)}}=
\sqrt{1-\alpha_x}\ket{\theta_{0x}}\otimes \ket{\psi_{0x}}
+\sqrt{\alpha_x}\ket{\theta_{1x}}\otimes \ket{\psi_{1x}}
\end{align*}
with $\alpha_x\in [0,1]$.
Since the encoding channel is applied to $G_1$ alone, the reduced state of $G_2$ remains unchanged. 
Thereby, the eigenvalues $(1-\alpha_x,\alpha_x)$ must be independent of $x$. 
That is, $\alpha_x\equiv \alpha$ for $x\in\mathcal{X}$, hence
\begin{align}
H(G_2|X)_\omega=h_2(\alpha) \,.
\label{Equation:Depolarizing_Input_Entropy}
\end{align}

Furthermore, the depolarizing channel is unitarily covariant, i.e., $\mathcal{N}(U\rho U^\dagger)=U \mathcal{N}(\rho )U^\dagger$ for every unitary $U$ on $\mathcal{H}_A$. Thus,
\begin{align}
H(B|X)_\omega&=H(\mathcal{N}(\widetilde{\phi}_{A}))
=h_2\left(\alpha*\frac{\varepsilon}{2} \right)
\end{align}
where $\ket{\widetilde{\phi}_{G_2 A}}= (1-\alpha)\ket{00}+\alpha\ket{11}$, and similarly,
\begin{align}
&H(G_2 B|X)_\omega=H\left((\mathrm{id}\otimes\mathcal{N})(\widetilde{\phi}_{G_2 A})\right)
\nonumber\\
&=H\left(\frac{\alpha \varepsilon}{2},\frac{(1-\alpha)\varepsilon}{2},
  \frac{1}{2}-\frac{\varepsilon}{4}\pm\sqrt{
  \frac{\varepsilon^2}{16}-(1-\alpha)\alpha \varepsilon\left(1-\frac{3\varepsilon}{4}\right)+\frac{1-\varepsilon}{4}
  } \right)
\label{Equation:Depolarizing_Converse_Last}
\end{align}
(see \cite{LeungWatrous:17p}).
As the output entropy is bounded by 
$H(B)_\omega\leq 1$,
the converse follows from 
\eqref{Equation:Depolarizing_Converse_1}-\eqref{Equation:Depolarizing_Converse_Last}.

Achievability follows as in \cite[Example 1]{PeregDeppeBoche:23p1}.
 Instead of a classical mixture, we now use quantum superposition. 
Set
$\ket{\phi_{G_1 G_2}}\equiv 
\sqrt{1-\alpha}\ket{00}+\sqrt{\alpha}\ket{11}
$, 
$p_X = \left( \frac{1}{2},\frac{1}{2} \right)$, 
$\mathcal{F}^{(x)}(\rho)\equiv \mathsf{X}^x \rho \mathsf{X}^x $, 
where $\mathsf{X}$ is the bitflip Pauli operator.
Thus,  $\alpha=0$ and $\alpha=\frac{1}{2}$ 
achieve the unassisted capacity and entanglement-assisted capacity, respectively. The resulting region is the set on the right-hand side of \eqref{Equation:Depolarizing_Region}.
\end{proof}

%

\section{Summary and Discussion}
We address communication over an entanglement-breaking  quantum channel, given \emph{unreliable} entanglement assistance. Previous work established a multi-letter asymptotic formula and presented the quantum ``superposition coding" (SPC) achievable region \cite{PeregDeppeBoche:23p1}.
Here, we show that the region is optimal for entanglement-breaking channels, and we single-letterize the formula, providing a complete characterization of the capacity region.
Furthermore, we derive a closed-form expression for the  qubit depolarizing channel, with a parameter $\varepsilon\geq \frac{2}{3}$. It is further demonstrated that the capacity region is strictly larger than the time-division rate region. From a networking perspective, this finding is nontrivial and highlights a quantum behavior arising from superposition.

We conclude with a discussion on the application in a dynamic communication network, the importance of single-letterization, the role of  entanglement breaking channels, and
the challenges posed by unreliable entanglement resources --- the underlying motivations, the concept of "hard decision decoding," the links to classical models, surprising behavior, and expected impact. 

\subsection{Dynamic Communication and Entanglement Resources}
In a dynamic communication network, information is not necessarily transmitted between  two particular nodes at every point in time. In principle, in the ``quiet" period of time, entanglement can be generated between those nodes. 
While entanglement can be harnessed to generate shared randomness, its potential utility extends far beyond that \cite{ChitambarGour:19p,BBDFFJS:21b}.
This motivates using entanglement to enhance various communication networks and applications, such as the Internet of Things (IoT) \cite{NotzelDiAdamo:20c1,NoetzelDiAdamo:20c2,NoetzelDiAdamo:20c,Noetzel:20p}.

Superdense coding \cite{BennetWiesner:92p} is a fundamental communication protocol, where a pair of classical bits is transmitted using just one instance of a noiseless qubit channel and a maximally entangled pair. This means that entanglement assistance effectively doubles the rate at which classical messages can be sent over a noiseless qubit channel.

\subsection{Single Letterization}
\label{Section:Single_Letterization}
In Shannon theory, the efficiency of communication across noisy channels is described by the concept of channel capacity. The capacity is defined as the maximum transmission rate that permits an error probability that tends to zero in the limit of an infinite blocklength.
Remarkably, Shannon \cite{Shannon:48p} proved that the capacity of a classical channel $W_{Y|X}$ admits a single-letter formula, i.e., a non-asysmtotic expression.
The significance of such single letterization is attributed to the following:
\begin{enumerate}
\item
\emph{Computability:}
Shannon's capacity formula is
generally considered to be ``easy to compute" in the sense that
given the channel statistics,  there are efficient algorithms, such as the Blahut-Arimoto algorithm \cite{Arimoto:72p,Blahut:72p},  
that can solve this convex optimization problem numerically,  up to a given precision and provided that the input and output
dimensions are not too
large. 
On the other hand, a multi-letter formula, of the form 
\begin{align}
\lim_{n\to \infty} \frac{1}{n} f(W^{\otimes n}) \,,
\end{align}
is difficult to compute since the dimensions of $W^{\otimes n}$ grow exponentially with $n$.

\item
\emph{Uniqueness:}
A multi-letter formula
does not
uniquely characterize the capacity of a channel for a given 
task \cite{Wilde:17b}. 
For instance, the capacity of a classical channel can be expressed as  \cite[Sec. 13.1.3]{Wilde:17b}
\begin{align}
\lim_{n\to \infty} \frac{1}{n} f_c(W^{\otimes n})=f_1(W)
\label{Equation:Not_Unique_Regularization}
\intertext{where}
f_c(W)=H(X)-cH(X|Y)
\end{align}
for every constant $c\geq 1$.
The multi-letter formulas 
 $\lim \frac{1}{n} f_1(W^{\otimes n})$ and 
 $\lim \frac{1}{n} f_5(W^{\otimes n})$ have a different form, and yet, both describe the channel capacity.
Hence, such a multi-letter description  is not unique.

\item
\emph{Optimal coding:}
Single-letter formulas provide valuable insights into optimal coding strategies in various settings.
For instance, the characterization for the multiple access channel captures coding techniques such as time sharing and successive-cancellation decoding \cite{Ahlswede:74p,WieseBoche:13p,MDCSPP:22p}. 
For parallel Gaussian channels, the capacity formula leads to the water filling power allocation 
\cite{Shannon:49p,BiglieriProadisShamai:98p,MaZhuZhang:23p}, 
among other applications.

\end{enumerate}

Unfortunately, a single-letter characterization for the capacity of a quantum channel is an open problem \cite{Holevo:12b,Pereg:22p}. 
Nevertheless, it is important to note that multi-letter characterizations remain significant  \cite[Remark 7]{PeregDeppeBoche:21p}. In many examples,  the capacity can be evaluated exactly based on the multi-letter result \cite{Holevo:12b,Wilde:17b}. 
Furthermore, there are interesting phenomena that can be observed even when a single-letter expression for the capacity is not available \cite{SmithYard:08p,PeregDeppeBoche:23a}.

In the entanglement-assisted communication setting, where
pre-shared entanglement resources are available to the transmitter and the receiver, 
 a complete single-letter characterization is well established, and can be viewed as the quantum parallel of Shannon's capacity theorem \cite{BennettShorSmolin:02p} \cite[Remark 5]{PeregDeppeBoche:21p}.
 The characterization with entanglement assistance provides a computable upper bound for
unassisted communication as well. 
 Therefore, entanglement-assisted communication has favorable attributes from both performance and analysis perspectives.

\subsection{Entanglement Breaking Channels}
Entanglement breaking is a fundamental property of a large class of quantum channels (see Section~\ref{Subsec:Notation}), including measurement channels and classical-quantum (c-q) channels \cite{HorodeckShorRuskai:03p}. The qubit depolarizing channel is  entanglement breaking if and only if the depolarization parameter is greater than or equal to $2/3$  \cite{MoravvcikovaZiman:10p}.

While an entanglement-breaking channel cannot be used to generate entanglement, it may facilitate the transmission of classical messages, and entanglement assistance can increase the channel  capacity for sending classical information substantially \cite{HaoShiLiShapiroZhuangZhang:21p}.
 Furthermore, the capacity without assistance is solved as well.
 Shor \cite{Shor:02p} established the single-letter characterization of the unassisted capacity by first showing that the Holevo information of an entanglement-breaking channel is additive.

 Here, we used a more direct approach, which was recently pointed out by
the author \cite{Pereg:22p} 
 (see also \cite{WangDasWilde:17p}),  proving a single-letter converse proof ``from scratch".
The proof is based on the representation of an entanglement breaking channel as a serial concatenation of a measurement channel and a c-q channel, along with the data-processing inequality (see capacity proof in Section~\ref{app:capacity}).

\subsection{Unreliable Entanglement Resources}
Communication with unreliable entanglement assistance  was proposed by
 the author, along with Deppe and Boche \cite{PeregDeppeBoche:23p1}.
 The model accounts for the practical challenges and low efficiency of entanglement generation in current implementations and experiments \cite{JaegerSergienko:14p}.
 The framework is inspired by 
classical approaches for unreliable cooperation resources in the classical literature
 \cite{Steinberg:14c,HuleihelSteinberg:17p,ItzhakSteinberg:17c,ItzhakSteinberg:21p,PeregSteinberg:21c4,OzarowShamaiWyner:94p,KarasikSimeoneShamai:13p}.
The focus in our quantum setting, however,  is on a point-to-point quantum channel and the reliability of correlation resources.

 Our principle of operation provides reliability by design, by  adapting the communication rate based on the availability of entanglement assistance, while eliminating the need for feedback, repetition, or distillation. As illustrated in the Introduction, and as opposed to other models in the literature \cite{ZhuangZhuShor:17p}, the availability of entanglement is not associated with a probabilistic model either. 
 Here, the receiver is aware of the availability or absence of entanglement resources through heralded entanglement generation.
 The receiver performs  ``hard decision decoding"  \cite{Proakis:01b}, deciding whether the entanglement resources are usable or not at all.
 
Drawing a parallel with the classical cooperation model
\cite{HuleihelSteinberg:17p},
the unreliable assistance model is based on the 
engineering aspects and the architecture of modern communication networks.
%
We anticipate that future quantum communication networks will adhere to similar reliability principles.
In particular, we envision that in a large quantum communication
network, the availability of entanglement resources will not be guaranteed in advance. 
Specifically, the accessibility of entanglement resources will depend on factors such as  weather conditions,  the operational  status of quantum repeaters, or  the willingness of peers 
to provide assistance. In such a network, the transmitter and the receiver are aware of the \emph{possibility} that entanglement assistance will be available,  yet its confirmation remains uncertain until reception.

The model  exhibits unexpected behavior, highlighting that    communication with unreliable entanglement assistance is not a mere combination of the entanglement-assisted and unassisted settings.
We have shown that for
an entanglement-breaking depolarizing channel,  quantum SPC outperforms time division even in this simple point-to-point setting. This is surprising because SPC is typically useful in more complex network setups, and does not yield an advantage in point-to-point communication.
For example,   in a classical broadcast channel with degraded messages, where a transmitter $X$ communicates with two receivers, $Y_1$ and $Y_2$, SPC is unnecessary when the receivers' outputs are identical, i.e., $Y_1=Y_2$,  as  
the capacity region can be attained using a simpler approach of time division. 
That is, concatenating two single-user codes is optimal.
In our context, the system can be regarded as a quantum broadcast channel with degraded messages where one receiver has entanglement assistance, and the other does not.  Nevertheless,  the output states of the receivers are identical (without violating the no-cloning theorem, as we consider two alternative scenarios).
The expectation would be that time division, combining assisted and unassisted codes, achieves optimality. However, this expectation is proven false as quantum SPC can outperform time division, based on the combination of a superposition code with a superposition state.

We expect that the present work will have a significant impact due to its relevance to practical systems, the interesting and unexpected properties, and  the potential applicability of our reliability principles across a wide variety of tasks and protocols that rely on pre-established entanglement. These range between research areas such as communication, distributed computing,  complexity theory, and cryptography,  among others.


\section*{Acknowledgments}
The author thanks Christoph Becher (Universit\"at des Saarlandes) for his observations on the practical implementation of heralded entanglement, as well as Quntao Zhuang (University of Southern California) and Saikat Guha (University of Arizona)
for useful discussions on entanglement-breaking channels.

This work was supported by the Israel VATAT Junior Faculty Program for Quantum Science and Technology through Grant 86636903, Israel Science Foundation (ISF), Grants 939/23 and 2691/23,  Chaya Career Advancement Chair, Grant 8776026,  German-Israeli Project Cooperation (DIP), Grant
2032991, and  Nevet Program of the Helen Diller Quantum Center at the Technion.


\appendix
\section{Convexity Properties}
\label{app:Convexity}
In Section~\ref{Subsection:Equivalence}, 
we presented convexity properties of the region $\mathcal{R}_{\text{EA}*}(\mathcal{N})$, as defined in \eqref{eq:calRClea}. 
Since the derivation is technical, we have delegated the proof to the appendix. 

\subsection{Proof of Lemma~\ref{Lemma:Convexity}}
Let $\lambda\in [0,1]$.
Consider two rate pairs, $(R_u,R_u')$, 
$u\in\{1,2\}$, that belong to the rate region $\mathcal{R}_\text{EA*}(\mathcal{N})$. Then, we have 
$R_u\leq I(X;B|U=u)_\omega$ and $R_u'\leq I(G_2^{(u)};B|X,U=u)_\omega$ for some conditional distribution $p_{X|U}$, entangled state $\varphi_{G_1^{(u)} G_2^{(u)}}$, and encoding channel $\mathcal{F}^{(x,u)}_{G_1^{(u)}\to A}$.

Consider the joint state
\begin{align}
\varphi_{\bar{A}_0 \bar{A}_1}=
 \varphi_{G_1^{(1)} G_2^{(1)}}
\otimes \varphi_{G_1^{(2)} G_2^{(2)}} \,.
\end{align}
Given $u\in\{1,2\}$, define an encoding channel 
$\bar{\mathcal{F}}^{(x,u)}_{\bar{A}_0 \to A}$
that maps from $\bar{A}_0$ to $A$,
 such that
\begin{align}
\bar{\mathcal{F}}_{\bar{A}_0\to A}^{(x,1)}&\equiv \mathcal{F}^{(x,1)}_{G_1^{(1)}\to A} \circ \trace_{G_1^{(2)}} \,,
\\
\bar{\mathcal{F}}_{\bar{A}_0\to A}^{(x,1)}&\equiv \mathcal{F}^{(x,1)}_{G_1^{(1)}\to A} \circ \trace_{G_1^{(2)}} \,.
\end{align}
The system $A$ is then sent through the channel
$\mathcal{N}_{A\to B}$.
We note that if $U=1$, then the output is uncorrelated with $G_2^{(2)}$.
Similarly, for $U=2$, there is no correlation with $G_2^{(1)}$.
Therefore,
\begin{align}
I(\bar{A}_1;B|X,U=u)_\omega=I(G_2^{(u)};B|X,U=u)_\omega
\end{align}
for $u\in\{1,2\}$.

Let $U\sim\mathrm{Bernoulli}(\lambda)$, with $\lambda\in [0,1]$.
  Observe that  the convex combinations of the rates satisfy
\begin{align}
R_\lambda&\equiv (1-\lambda)R_1+\lambda R_2 \leq
I(X;B|U)_\omega\leq I(XU;B)_\omega
\end{align}
and
\begin{align}
R_\lambda'&\equiv (1-\lambda)R_1'+\lambda R_2' \leq
I(\bar{A}_1;B|XU)_\omega 
\end{align}
with respect to the following states,
\begin{align}
\omega_{XU \bar{A}_1  A}
&=
\sum_{(x,u)\in\mathcal{X}\times\mathcal{U}} p_U(u) p_{X|U}(x|u) \ketbra{x,u}
\nonumber\\&
\otimes ( \mathrm{id}\otimes \bar{\mathcal{F}}^{(x,u)}_{\bar{A}_0\to A})
(\varphi_{ \bar{A}_1  \bar{A}_0}) \,,  \\
\omega_{XU \bar{A}_1  B}
&=
(\mathrm{id}\otimes\mathcal{N}_{A\rightarrow B})(\omega_{XU\bar{A}_1  A})
\,.
\end{align}
As we substitute
 the auxiliary variable $\bar{X}\equiv (X,U)$ in \eqref{eq:calRClea}, we observe that the pair $(R_\lambda,R_\lambda')$ is in the rate region $\mathcal{R}_\text{EA*}(\mathcal{N})$ as well. Thereby, the region is a convex set.
\qed

\subsection{Proof of Corollary~\ref{Corollary:Region_lambda}}
The proof for Corollary~\ref{Corollary:Region_lambda} follows from the convexity property in  Lemma~\ref{Lemma:Convexity}.
It suffices to consider the boundaries of the two regions in \eqref{eq:calRClea3}.

Consider $(R_1,R_1')= (I(X;B)_\omega,I(G_2;B|X)_\omega)$.
Next, we claim that the rate pair
$(R_2,R_2')=(0,I(XG_2;B)_\omega)$ belongs to 
$\mathcal{R}_\text{EA*}(\mathcal{N})$ as well. To see this, set
\begin{align}
&
\widetilde{X}\equiv \emptyset
\,,\;
\widetilde{A}_1\equiv (X,G_2) 
\,,\;
\widetilde{A}_0\equiv A
\,,\; \text{ and } 
\nonumber\\
&
\varphi_{\widetilde{A}_1 \widetilde{A}_0}\equiv \omega_{XG_2 A} 
\,.
\end{align}
As for the convex combination of $(R_1,R_1')$ and $(R_2,R_2')$, we have
\begin{align}
R_\lambda&\equiv (1-\lambda)R_1+\lambda R_2 =(1-\lambda)I(X;B)_\omega
\label{Equation:R_lambda2}
\intertext{and}
R_\lambda'&\equiv (1-\lambda)R_1'+\lambda R_2'
\nonumber\\
&= (1-\lambda)I(G_2;B|X)_\omega+\lambda I(XG_2;B)_\omega
\nonumber\\
&= I(G_2;B|X)_\omega+\lambda[ I(XG_2;B)_\omega-I(G_2;B|X)_\omega]
\nonumber\\
&= I(G_2;B|X)_\omega+\lambda  I(X;B)_\omega
\label{Equation:R_lambda_Prime2}
\end{align}
by the chain rule for the quantum mutual information. By Lemma~\ref{Lemma:Convexity}, the pair $(R_\lambda,R_\lambda')$ belongs to the region $\mathcal{R}_{\text{EA}*}(\mathcal{N})$, hence the corollary follows from \eqref{Equation:R_lambda2}-\eqref{Equation:R_lambda_Prime2}.
\qed


\ifdefined\bibstar\else\newcommand{\bibstar}[1]{}\fi

\end{document}